\documentclass[runningheads]{llncs}
\usepackage{graphicx}
\usepackage[utf8]{inputenc}
\usepackage[english]{babel}
\usepackage{authblk}
\usepackage{amsmath}
\usepackage{amssymb}
\usepackage{titling}
\usepackage[title]{appendix}
\usepackage{float}
\usepackage{mathtools}
\usepackage{tcolorbox}
\usepackage{tabularx}
\usepackage{dsfont}
\usepackage[linesnumbered,ruled,vlined]{algorithm2e} 
\usepackage{tikz}
\usetikzlibrary{arrows,calc}
\usepackage[style=numeric]{biblatex}
\usepackage{scalerel,stackengine}
\stackMath
\newcommand\reallywidehat[1]{%
\savestack{\tmpbox}{\stretchto{%
  \scaleto{%
    \scalerel*[\widthof{\ensuremath{#1}}]{\kern-.6pt\bigwedge\kern-.6pt}%
    {\rule[-\textheight/2]{1ex}{\textheight}}
  }{\textheight}%
}{0.5ex}}%
\stackon[1pt]{#1}{\tmpbox}%
}
\newcommand*{\mathfn}{\fontfamily{cmss}\selectfont}
\newcounter{protocol}

\title{Trustless, privacy-preserving blockchain bridges}
\author{Drew Stone}
\affil{\textit{drew@webb.tools}}
\date{}

\begin{document}
\maketitle
\begin{abstract}
     In this paper, we present a protocol for facilitating trust-less cross-chain cryptocurrency transfers that preserve privacy of bridge withdrawals. We leverage zero-knowledge primitives that are commonly used to design cryptocurrency mixing protocols to provide similar functionality but across two or more blockchains. To that end, we receive cryptocurrency mixing "for free" through the bridge operations and describe how to extend these protocols to incentivise bridge transfers using ideas from \cite{le2020amrautonomous}. We describe how resulting protocols lead to similar \textit{vampire} attacks coined in the Uniswap vs. Sushiswap saga but across chains and discuss incentivisation mechanisms.
\end{abstract}

\section{Introduction}
Cross-chain bridges are popular mechanisms for transferring native cryptocurrency assets to new blockchains. Oftentimes, bridges serve to bring new liquidity in the form of new cryptocurrency assets to a new or existing blockchain protocol. Other attractions for using cross-chain bridges are to utilise certain cryptocurrencies with more expressive scripting engines, such as utilising bridged Bitcoin \cite{nakamoto2012bitcoin} assets on Ethereum \cite{wood2014ethereum} in various smart contracts. The range of applications utilising bridged assets ranges from simple exchange between new currency pairs to various liquidity mining schemes and more.

The mechanisms for building bridges exist in both centralized (\textit{custodial}) and decentralized (\textit{non-custodial}) forms. Custodial systems however are objectively worse in theory albeit simpler to implement in practice because a user entrusts their cryptocurrency with another, potentially malicious, entity. For that purpose, it is ever more crucial to build secure and robust, decentralized blockchain bridges to transfer cryptocurrencies across ledgers.

The main ingredient for building decentralized blockchain bridges is a \textit{light-client}, a system that verifies proofs of work of other blockchains. If both blockchains being bridged can support the existence of a light-clients for each respective chain, then a bridge can be built to transfer cryptocurrencies between them.

\subsection{Coin Mixers}
Cryptocurrency mixers have been popular since the creation of Bitcoin \cite{nakamoto2012bitcoin}. A cryptocurrency mixer is a system that allows holders of cryptocurrency to anonymize their cryptocurrency holdings by depositing fixed sized amounts and withdrawing fixed sized amounts to new addresses such that linking these actions is considered hard. The goal of a cryptocurrency mixer is to anonymize any link between deposits (inputs) and withdrawals (outputs),

A mixer is only as good as its anonymity set, defined as the set of potential deposits a withdrawal could link to. It is crucial to grow the anonymity set in practice to achieve large levels of privacy.

\subsection{DeFi}
The rise of decentralized finance gives us interesting mechanisms for bootstrapping large anonymity sets. DeFi "farming" schemes present new ways for incentivising users to participate in new cryptocurrency protocols by giving early users governance tokens that offers voting rights and/or direct cash-flows within the protocol. We cast the design of private bridge protocols in a similar light, distributing a useless tokens to participants of the bridge based on certain criteria such as the time one remains locked in one side of the bridge.
 
 \subsection{Our Contributions}
 In this paper, we combine these three concepts, \textit{light clients}, \textit{mixers}, and DeFi incentivisation techniques to build a trustless and private cryptocurrency bridge. We outline our protocol in the form of a smart contract that manages the state and functionality necessary for both a light client and a bridge. We achieve private bridge transfers by utilising succinct zero-knowledge proofs to verify the validity of withdrawals and prevent double-spending. While we utilise light clients within our system, we defer formalisation of light client work to previous research and focus primarily on the bridge mechanism.
 
 \subsection{Related Work}
 \textbf{Light clients} are well studied mechanisms for making blockchains interoperable. One main crtieria for designing light clients is ensuring they are trustless, through the ability to verify proofs of work and consensus of other blockchains. Light clients are required to store data, in the form of block headers, of these other blockchains. Over time, the storage costs of these headers because non-negligable and so more thoughtful designs are necessary. Proposals such as NiPoPoWs \cite{cryptoeprint:2017:963}, FlyClients \cite{cryptoeprint:2019:226}, and Plumo \cite{gabizon2020plumo} present ways of optimising light-client designs using zero-knowledge proofs. We draw inspiration from these designs to design a protocol for our specific application.

\noindent 
\textbf{Privacy in cryptocurrencies} is well-studied from both a mechanism design and analysis perspective.  Much work has gone into designing private and scalable cryptocurrency protocols using zero-knowledge proofs, starting from the original ZeroCash work \cite{sasson2014zerocash} to the initial Monero work \cite{cryptoeprint:2015:1098} which utilises linkable ring signatures to facilitate a cryptocurrency protocol. Cryptocurrency mixers on he other hand present ways of adding privacy to blockchain protocols without such functionality baked in. Prposals such as Zether \cite{bunz2019zether} aim to build ZeroCash-like functionality on Ethereum \cite{wood2014ethereum} in a smart contract. Similar in spirit, we design our protocol on muliple smart contracts.

An example of a cryptocurrency mixer-like system on Bitcoin \cite{nakamoto2012bitcoin} is the CoinJoin protocol \cite{coinjoin} for Bitcoin and any UTXO based blockchain. CoinJoin allows users to combine many transaction inputs and outputs within a single transaction to obfuscate the flow of funds.
\section{Preliminaries}
We denote by $1^\lambda$ the security parameter of our protocols and {\mathfn negl}$(\lambda)$ a negligable function in $\lambda$. We denote by $\cdot || \cdot$ the concatenation of two elements, in binary string representations.

We will restrict our focus to operating over a prime, finite-field $\mathbb{F}_p$ for prime number $p\in \mathbb{Z}$. In our protocols our prime number $p$ is on the order of $\lambda$ bits.  We denote PPT as polynomial, probabilistic time.

\subsection{Hash functions}
We let $H: \{0,1\}^*\rightarrow \mathbb{F}_p$ be a cryptographic hash function that maps binary strings to elements in $\mathbb{F}_p$. We denote by $\xleftarrow[]{\text{\$}}$ a random sampling from a set.

\begin{definition}
A family of cryptographic hash functions $\mathcal{H}$ is called \textbf{collision-resistant} if $\forall$ PPT adversaries $\mathcal{A}$, given $H\xleftarrow[]{\text{\$}}\mathcal{H}$, the probability that $\mathcal{A}$ finds $x,y$ such that $H(x)=H(y)$ is {\mathfn negl}$(\lambda)$.
\end{definition}

\subsection{zkSNARKs}
A zero-knowledge Succinct Non-interactive Argument of Knowledge (zkSNARK) is a computation proof of an NP relation $\mathcal{R}$ that allows a prover $\mathcal{P}$ the ability to demonstrate knowledge of a witness $w$ for a statement $x$ to a verifier $\mathcal{V}$ without disclosing any information about the witness $w$. We aim to use a zkSNARK protocol that achieves \textit{zero-knowledge}, \textit{succinctness}, \textit{soundness}, and \textit{completeness}.
\begin{itemize}
    \item \textbf{\textit{Zero-knowledge}}: Proofs $\pi$ generated by $\mathcal{P}$ for $(x,w)\in\mathcal{R}$ disclose no information about the witness $w$.
    \item \textbf{\textit{Succinctness}}: Proofs $\pi$ generated by $\mathcal{P}$ for $(x,w)\in\mathcal{R}$ are small and do not scale with the complexity of $\mathcal{R}$.
    \item \textbf{\textit{Soundness}}: If $\mathcal{V}$ accepts a proof $\pi$ generated by $\mathcal{P}$, with probability $1-\text{\mathfn negl}(\lambda)$ $\mathcal{P}$ knows $(x,w)\in\mathcal{R}$.
    \item \textbf{\textit{Completeness}}:  If $\mathcal{P}$ knows $(x,w)\in\mathcal{R}$ then they can generate a proof $\pi$ that $\mathcal{V}$ accepts with probability $1-\text{\mathfn negl}(\lambda)$.
\end{itemize}

We utilise zkSNARKs for arithmetic circuits outlined in the original ZeroCash \cite{sasson2014zerocash} protocol and \cite{le2020amrautonomous}. Transformations from NP relations $\mathcal{R}$ to arithmetic circuits allow an individual to encode circuits whose satisfiability rests on knowing elements of $\mathcal{R}$. We denote the circuit generated by $\mathcal{R}$ as $C$.

Statements $x$ in our relation $\mathcal{R}$ have both public and private inputs due to the ability to hide parts of a satisfying witness $w$ for $x$. We will describe public parameters and public inputs for our zero-knowledge proof system using this $\text{\mathfn font}$.

\begin{definition}
A zkSNARK protocol for an arithmetic circuit $C$ is composed of a trio of algorithms: (\textsc{Setup, Prove, Verify}).
\begin{itemize}
    \item {\mathfn (pp) }$\longleftarrow {\textsc{Setup}}(1^\lambda, C)$ generates a set of public parameters $\text{{\mathfn pp}}\in\mathbb{F}_p^k$ where $k$ is a parameter derived from the circuit $C$.
    \item $\pi\longleftarrow {\textsc{Prove}(\text{\mathfn pp}, x, w)}$ generates a proof proof $\pi$ using a statement $x$ and witness $w$.
    \item $b\longleftarrow {\textsc{Verify}(\text{\mathfn pp}, x, \pi)}$ outputs a boolean value $b\in\{0,1\}$ depending on if the proof $\pi$ is satisfying for the statement $x$.
\end{itemize}
\end{definition}
We require this protocol to provide \textit{soundness}, \textit{completeness}, and \textit{zero-knowledge} as well as \textit{succinctness} and \textit{simulation-extractability}. We refer the reader to \cite{groth2016size} for formal explanations of these properties.

\subsection{Merkle Trees}
We utilise binary, merkle trees as the foundation of our mixing protocols. We denote preimages of leaves in our merkle tree as $z\in\{0,1\}^*$. Leaves in the tree are elements $y=H(z)\in\mathbb{F}_p$ for a collision-resistant hash function $H\xleftarrow[]{\$}\mathcal{H}$. Building the tree follows the standard algorithm of repeatedly hashing pairs of child elements until a single element, the merkle root, remains.

We denote by $I$ an instance of a merkle tree and its merkle root $\text{\mathfn MR}_I$ and refer to the state of the merkle root and instance at time $t$ as {\mathfn MR}$^t_I$ and $I^t$. We will drop the $t$ when the context is clear. We define a function $\text{\textsc{path}}_I(y)$ to be the merkle proof path for a leaf $y$.
\begin{definition}
A merkle tree instance $I$ is defined with the algorithms {\textsc{(Setup, Add, Verify, Prove, Verify\_Snark)}}.
\begin{itemize}
    \item $(\text{\mathfn pp}, I)\longleftarrow\textsc{Setup}(h, \lambda, C)$ generates a merkle tree instance $I$ with height $h$ and public paramters for a merkle tree root reconstruction circuit $C$ with security parameter $\lambda$. For this, we run \textsc{ZK.Setup}$(\lambda, C)$ from our zkSNARK protocol as a subroutine.
    \item $b\longleftarrow\textsc{Add}(y, I)$ adds a leaf $y\in \mathbb{F}_p$ to the merkle tree instance $I$, returning a bit for failure/success $b\in \{0,1\}$.
    \item $b\longleftarrow\textsc{Verify}(y,\textsc{path}_{I^t}(y), I^t)$ verifies a merkle path proof for a leaf $y$ in a merkle tree instance $I^t$ by re-computing a merkle root $\text{\mathfn {MR}}_B$ and returning a bit for failure/success $b\in \{0,1\}$ if $\text{\mathfn {MR}}_B = \text{\mathfn MR}^t_{I}$.
    \item $\pi_I\longleftarrow\textsc{Prove}( \text{\mathfn pp}, \text{\mathfn sn}, r, s,I)$ generates a proof $\pi_I$ for the following relation:
    \[
      \mathcal{R}_{\text{\mathfn MR}_I} = \left\{ (x,w)\ \middle\vert \begin{array}{l}
        w=(r,s,\textsc{path}_I(H(r||s)))~\land~\text{\mathfn sn} = H(r) \\
        x := \textsc{Verify}(H(r||s),\textsc{path}(H(r||s), \text{\mathfn MR}_I)=1  
      \end{array}\right\}
    \]
    We generate a proof for this relation using $\textsc{ZK.Prove}(\text{\mathfn pp}, x, w, \text{\mathfn MR}^t_I)$ as a subroutine. {\mathfn MR}$^t_I$ is a public input, $w$ is a private input.
    \item $b\longleftarrow\textsc{Verify\_Snark}(\text{\mathfn pp}, \text{\mathfn sn},\pi_I,I)$ verifies a zero-knowledge proof $\pi_I$ by running \textsc{ZK.Verify}$(\text{\mathfn pp},x,\pi_I, \text{\mathfn MR}_I, \text{\mathfn sn})$ from our zkSNARK protocol as a subroutine. A merkle root {\mathfn MR}$_I$ and nullifier {\mathfn sn} are provided as public inputs to the verifier.
\end{itemize}
\end{definition}

\noindent\textbf{Note:} We will eventually abuse the notation for \textsc{Prove} and \textsc{Verify\_Snark} to satisfy a zkSNARK that verifies the statement representing the OR relation of $\mathcal{R}_I$ but with 2 merkle roots of merkle tree instances, $\mathcal{R}_{\text{\mathfn MR}_A} \cup \mathcal{R}_{\text{\mathfn MR}_B}$. For this we add extra inputs to our functional definitions.
\begin{proposition}
The (\textsc{Setup, Prove, Verify\_Snark}) algorithms define a zkSNARK protocol that preserves \textit{soundness, completeness, zero-knowledge, succinctness, and witness-extractability} properties.
\end{proposition}
\begin{proof}
It follows from \cite{sasson2014zerocash} that we can build an arithmetic circuit for the relation $R$ as well as generate succinct proofs that preserve \textit{soundness}, \textit{completeness}, \textit{zero-knowledge}, and which are \textit{simulation-extractable} as long as our collision-resistant hash function $H$ admits a circuit that has polynomially many constraints and is polynomial-time verifiable. This is not a strict requirement as there exists many so called SNARK-friendly hash functions with low arithmetic complexity defined over prime fields, such as MiMC \cite{albrecht2016mimc} and Poseidon \cite{grassi2020poseidon}. We refer to these papers for proofs of their polynomial time verifiability.
\end{proof}

\subsection{Smart contracts}
A smart contract is a piece of software that is executed within the execution context of a blockchain protocol. A smart contract is considered "smart" because it allows for generalised computational that is often turing-complete. For background on smart contract blockchains, refer to Ethereum \cite{wood2014ethereum}, which was the first to implement these primitives.

For our purposes, we assume smart contracts are stateful programs with $O(1)$ read/write access. We can instantiate a variety of algorithms on these smart contracts and interact with them through a peer-to-peer network.

\section{Model}
We restrict our attention to two smart-contract blockchain protocols $A$ and $B$ and a private bridge between them referenced by a smart contract $S$. Extending this to $n>2$ chains requires replicating $S$ on more chains in a similar fashion as is described. $S$ exists on both $A$ and $B$ but may manage different states, we denote these differences by $S_A$ and $S_B$. We will also omit these subscripts when we discuss shared functionality. To that end, we require that the smart-contract engines for $A$ and $B$ are \textit{compatible}; they support the same cryptographic primitives and computational programs necessary for deploying and using $S$. For the rest of the paper, we will assume that $A$ and $B$ are built using the same blockchain platform.

We require the existence of light-clients on $A$ and $B$ of one another, such that an external user can verify proofs of transactions and state of blockchain $B$ on $A$ and vice versa with high probability of success $1-\text{\mathfn negl}(\lambda)$. These light-clients have a delay $D$ until which transactions and state can be considered $safe$ by users on the bridged chain. Our delay $D$ is bounded in the presence of a single honest light-client relayer who cannot be censored. $S$ supports this functionality through the algorithms $\textsc{(Setup, Add\_Header, Add\_Bridge\_State)}$. We will describe these algorithms in later sections.

The smart contract $S$ supports algorithms $\textsc{(Setup, Deposit, Withdraw)}$ in order to operate the bridge. $S$ is also capable of owning a non-zero balance of cryptocurrency of each blockchain. In a typical blockchain environment, users have balances and spend them. We also assume that $S$ has the functionality to own a balance and spend tokens on blockchains $A$ and $B$. This functionality will be encoded into $S$.

\section{Protocol}
We let $S$ be a smart contract that supports a light-client functionality and bridge functionality. Light client functionality means that $S$ can process and verify proofs of work or finality of other blockchains, in order to track their consensus states. In a bridged network with $S$, light-client relayers share proofs of work as well as information about bridge operations -- deposits and withdrawals -- for $A,S_A$ on $S_B$ and vice versa.

$S$ uses two types of bridge related state. The first is a merkle tree instance $I$ with fixed depth or height $h$. This tree supports the API defined above, with the exception that \textsc{Prove} does not need to exist on $S$. The second type are arrays, which we will use to store the serial numbers {\mathfn sn} supplied when a user verifies a zero-knowledge proof on $S_A$ and $S_B$ through relaying state back and forth. We also have some state for handling light-client functionality. Each of these data structures are initialised as empty. We denote the cryptocurrency being sent over our bridge as $\$T$, without loss of generality $\$T$ can be native to either blockchain.

With this simplistic infrastructure deployed and operational, the high-level protocol works as follows: A user $u$ decides to transfer a non-zero number of tokens $\$T$ to blockchain $B$.

\begin{enumerate}
    \itemsep0em 
    \item A mechanism designer constructs the circuit $C$ for merkle tree membership verification and selects the hash function $H\xleftarrow[]{\$}\mathcal{H}$.
    \item $S$ is deployed to both $A$ and $B$. We generate public parameters {\mathfn pp} for both using $\text{\mathfn MT}.\textsc{Setup}(h,\lambda, C)$ and setup $S$ with a genesis state for the light-client.
    \item $u$ generates random values $r,s\xleftarrow[]{\$}\mathbb{F}_p$
    \item $u$ interacts with $S_A$ and deposits a fixed sized deposit $d$ of $\$T$ into $S_A$ and appends a merkle leaf $z$ to $I_A$, the merkle tree instance on $A$, using \textsc{Add}$(z,I_A)$.
    \item $u$ generates a proof $\pi_{I_A}$ of their deposit using $\text{\mathfn MT}.\textsc{Prove}( \text{\mathfn pp}, \text{\mathfn sn}, r, s,I_A, I_B)$.
    \item $u$ may submit proofs to withdraw their tokens after a delay $D'$.
    \item $u$ withdraws from $S_B$, minting a \textit{wrapped} $\$T$ on $B$, by submitting the proof generated previously, using $\text{\mathfn MT}.\textsc{Verify\_Snark}(\text{\mathfn pp}, \text{\mathfn sn},\pi_I,\text{\mathfn MR}_A, \text{\mathfn MR}_B)$.
\end{enumerate}

With this informal description we now define the smart contract facilitating this protocol and prove relevant properties about it.

\subsection{Smart contract design}
\begin{tcolorbox}[fonttitle=\bfseries, title=Bridge Contract $S_A$ for blockchain]
\textit{// Bridge related state variables} \\
\textbf{Merkletree} $I_A$ \\
$\mathbb{F}_p[~]$ {\mathfn nullifiers} \\
$\{0,1\}^{256}[~]$ {\mathfn A\_roots} \\
$\{0,1\}^{256}[~]$ {\mathfn B\_roots} \\ \\
\textit{// Light client related state variables} \\
$\{0,1\}^*[~]$ {\mathfn B\_headers} \\
$\{0,1\}$ {\mathfn initialised} = \textbf{false} \\ \\
\textsc{\underline{Setup}}$(B_0,h,\lambda,C)$ \\
\textsc{\underline{Add\_Header}}$(B_i,\phi)$ \\
\textsc{\underline{Add\_Bridge\_State}}$(\mathcal{C}, \mathcal{S},k)$ \\
\textsc{\underline{Deposit}}$(v,z)$ \\
\textsc{\underline{Withdraw}}($\pi_I, \text{\mathfn MR}^{t}_A, \text{\mathfn MR}^{t'}_B, \text{\mathfn sn}$) \\
\end{tcolorbox}

\begin{definition}
We let {\mathfn nullifiers} be a set of public nullifiers exposed on $S_A$ and relayed to $S_A$ by the light-client relayers when new proofs are submitted for verification on $S_B$.
\end{definition}
\begin{definition}
We let {\mathfn A\_roots} be an array of roots that represent intermediate states of $I_A$. Elements of this collection are denoted by {\mathfn MR}$^t_{A}$ for an arbitrary index $t$.
\end{definition}
\begin{definition}
We let {\mathfn B\_roots} be an array of roots that represent intermediate states of $I_B$. Elements of this collection are denoted by {\mathfn MR}$^t_{B}$ for an arbitrary index $t$.
\end{definition}
\begin{definition}
We let {\mathfn B\_roots} be an array of headers of $B$ that have been verified by the light client functionality.
\end{definition}

\noindent Descriptions of our interface follow:
\begin{itemize}
    \item \textsc{\underline{Setup}}$(B_0,h,\lambda,C)$ - Calls {\mathfn MT.}\textsc{Setup}$(h, \lambda, C)$ and initialises the first block header $B$ by adding to {\mathfn B\_headers} and {\mathfn B\_roots}.
    \item \textsc{\underline{Add\_Header}}$(B_i,\phi)$ - A function verifying a proof of a proof of work $\varphi$ for a new block header $B_i$.
    \item \textsc{\underline{Add\_Bridge\_State}}$(\mathcal{C}, \mathcal{S}, k)$ - A function to verify updates to $S_B$'s merkle tree and add valid roots to {\mathfn B\_roots}. This uses the header stored at {\mathfn B\_headers}$[k]$ when the light client is initialised.
    \item \textsc{\underline{Deposit}}$(v,z)$ - Validates the deposit $v$ and calls {\mathfn MT}.\textsc{Add}$(z,I)$. This function deducts a balance of $v$ coins from the sender.
    \item \textsc{\underline{Withdraw}}($\pi_I, \text{\mathfn MR}_A, \text{\mathfn MR}_B, \text{\mathfn sn}$) - Validates the merkle roots are valid, i.e. that they exist in the respective collections {\mathfn A\_roots, B\_roots} and also validates that $\text{\mathfn sn}\notin \text{\mathfn nullifiers}$. Then, it calls $\text{\mathfn MT}.\textsc{Verify\_Snark}$ with the arguments $(\text{\mathfn pp}_I, \text{\mathfn sn},\pi_I,\text{\mathfn MR}_A, \text{\mathfn MR}_B)$. If it succeeds, we withdraw funds to the sender.
\end{itemize}

\begin{proposition}
There exists a time $D'\geq D$ such that a user can transfer tokens from $S_A$ to $S_B.$ and withdraw them successfully on $S_B$.
\end{proposition}
\begin{proof}
Once $u$ successfully executes \textsc{Deposit} at time $t$ on $S_A$, we know that in time $D$, the merkle root $\text{\mathfn MR}^{t+D}_A$ will be recorded in the array {\mathfn A\_roots}. Thus at time $t+D$, $u$ can generate a proof $\pi$ for the membership of their deposit in either the root of the native merkle instance OR one collected in the auxiliary storage. Due to the completeness of our protocol, $u$ can successfully convince the verifier on the smart contract function \text{Verify\_Snark} with probably at least $1-\text{\mathfn negl}(\lambda)$.
\end{proof}
\subsection{Withdrawals}
Once a user $u$ deposits value into the tree, it is imperative that they cannot double-withdraw from the system, a problem similar to the double spending problem. We prevent this with the following observation and modification.
\begin{itemize}
    \item \textbf{Observation}: A user $u$ can only double spend if the nullifiers {\mathfn sn} used have not been relayed to the bridged blockchain. Since we assume relayed information arrives after delay $D$, a user $u$ would need to submit both withdrawals within $D$.
    \item \textbf{Modification \#1}: All withdrawals suffer from a required delay of $D + O(1)$ in processing time, defined below. When a user submits the same nullifier to both instances of $S$, the system will find out about the respective action in time $D$, which is enough to cancel the withdrawal when detected.
    \item \textbf{Modification \#2}: When a relayer updates the nullifier set, we require that each new nullifier added is checked for existence on the contract already to find these duplicates mentioned above.
\end{itemize}
\begin{proposition}
With the modifications above, a user $u$ cannot double withdraw tokens.
\end{proposition}
\begin{proof}
Due to the soundness requirement, we don't worry about proof forgery as that occurs with negligable probability. Rather we restrict our attention to whether any race condition exists in allowing double withdrawals.

Suppose a user $u$ submits 2 withdrawals at time $t$ using the same serial number {\mathfn sn}. This means at time $t$ the nullifier is exposed on both $S_A$ and $S_B$ and appended to its respective collection. At time $t+D$, both $S_A$ and $S_B$ will have learned of new nullifiers exposed on each respective contract. When checking existence for each new update, there will be a duplicate at {\mathfn sn}, causing both withdrawals to be canceled since an equal action is taken on the other smart contract.

Without loss of generality, suppose $u$ submits the withdrawal at time $t$ on $A$ and waits $t'$ before submitting on $B$. {\mathfn sn} will be recorded on $A$ and when the withdrawal on $S_B$ is relayed at $t+t'+D$, it will be cancelled on $S_A$. Similarly, when {\mathfn sn} is relayed to $S_B$ at time $t+D$, it will result in a verifiable duplicate on $S_B$. Both deposits will be cancelled. 
\end{proof}
\begin{proposition}
A processing delay of $D+\epsilon$ for $\epsilon = O(1)$ is sufficient to secure the system against double withdrawals.
\end{proposition}
\begin{proof}
We remark that if the time needed to check existence of a value in the array data structure is not $O(1)$ then we will switch it for a data structure with $O(1)$ search lookups. From above, we showed that if withdrawals were delayed by $D$ time then at least information about a double withdrawal will be transmitted through the relayer system. The earliest withdrawal transaction would be cancelled, followed by the later withdrawal being cancelled. Thus any additional constant $\epsilon$ of time is sufficient for delaying the possibility of not seeing the necessary information to detect a double withdrawal.
\end{proof}
\subsection{Storage Complexity}
Storing data for each of the necessary components is not for free. The storage complexity of our contract is on the order of the light-client storage complexity and the bridge functionality, which grow linearly in the presence of new information. However, light-client headers are likely to be larger in size that nullifiers and merkle roots, therefore the storage complexity $O(S)=O(\text{a light client})$ is on the order of the light client's storage requirements.

There are potential ways of reducing storage however, since blockchains are resource bound and optimisations are worth implementing in production. We describe future work at the end of the paper.
\subsection{A free mixer}
In our protocol, we describe a modification to the zkSNARK described in the merkle tree algorithms, specifically one that takes as input 2 merkle roots. A keen reader might notice that proving the OR relation of membership, in the native blockchain or in the target blockchain, generalises the capabilities for users with support for the functionailty of a mixer as well.
Users who deposit on $S_A$ can also withdraw from $S_A$ or $S_B$, and they cannot reuse a nullifer for this action between these nullifiers are tracked in a growing list {\mathfn nullifiers} that is kept up to date with delay $D$ on both instances of $S$.

Why is this a good thing? Users of blockchains $A$ or $B$ utilising $\$T$ can utilise the services of the bridge mechanism for mixer purposes specifically, effectively creating a larger anonymity set for all the users combined. This is due to the fact that our proofs are proofs over a combination of merkle tree accumulators rather than a single one.
\section{Incentives}
Recent work from \cite{le2020amrautonomous} as well as Tornado Cash \cite{TornadoC45:online} defined various schemes for incentivising liquidity into mixers. In this paper, we present an additional mechanism we could layer similar incentives onto. By participating on both sides of the bridge, users would earn tokens, potentially to shielded addresses, where rewards scale with the length they remained locked in their respective side of the bridge. We document a naive mechanism for issuing rewards below that is not optimal for privacy and discuss improvements thereafter.

\subsection{A naive approach}
The naive approach to rewards is to reward users for locking for at least a time $t$. A user can prove that their deposit has been in the system for that time by proving the deposit is accumulated in an older merkle root, one created before $t$ time from current root.

This approach may limit the privacy of the system in practice, since users would likely open themselves up to fingerprinting attacks using statistical timing techniques. As is referenced in the Tornado Cash anonymity mining scheme, it is thus necessary to reward users where these parameters are kept as private inputs to a zkSNARK circuit. Additionally, we would require that not only the time locked remains private but also the amount earned, such that rewards are issued to shielded addresses. We leave further investigation into this schemes to future work.

\subsection{Vampire Attacks}
A vampire attack is an homage to a vampire sucking blood from its victim. These bridge mechanism can be used in a similar fashion, to suck liquidity onto new blockchains. Since users are rewarded for providing liquidity on each side of the bridge in order to facilitate private transfers, users are incentivised to lock up tokens on both sides, causing liquidity to distribute across the bridge.

\section{Conclusion}
\subsection{Discussion}
In this paper, we outlined a mechanism for executing privacy-preserving bridge transfers by utilising 2 mixer-like mechanisms on both blockchains of the bridge. We briefly discussed the storage concerns and incentives of using this mechanism by rewarding an individual with a governance token that scales with their time locked on each side of the bridge.

\subsection{Future Directions}
It is worth considering optimisations to the storage complexity of the system. While we leave proofs of this to future work, we note that non-membership proofs present interesting tools for proving an element's non-membership in an accumulated value. One can consider merkle-izing the set of nullifiers and verfiably tracking the latest root of all reported nullifiers. A user's withdraw proof could then additionally prove that a serial number {\mathfn sn} is not accumulated in {\mathfn nullifiers} merkle root, eliminating the need to store a potentially large collection of data. \\ \\
Additonally, there is interesting incentive work to be done. For example, when rewards are different on each side of the bridge, liquidity might flow towards these incentives. The effects of this scheme can be investigated to understand how liquidity of cryptocurrency assets across different blockchains is affected.

\printbibliography

\end{document}